\documentclass[11pt,english,journal,onecolumn]{IEEEtran}
\usepackage[T1]{fontenc}
\usepackage[latin9]{inputenc}
\usepackage{geometry}
\usepackage{amsmath}
\usepackage{xpatch}
\usepackage{amssymb}
\usepackage{esint}
\usepackage{mathtools}
\usepackage{amsthm}
\usepackage{nicefrac}
\usepackage{gensymb}
\usepackage{bigints}
\usepackage{physics}
\usepackage{calligra}
\usepackage{graphicx}
\usepackage{epstopdf}
\usepackage{dsfont}
\usepackage{array,ragged2e}
\usepackage{enumitem}
\usepackage{etoolbox}
\usepackage{babel}
\usepackage[nopar]{lipsum}
\usepackage{psfrag}
\usepackage{setspace}

\newtheoremstyle{plain}
  {\topsep}   
  {\topsep}   
  {\itshape}  
  {0pt}       
  {\bfseries} 
  {.}         
  {5pt plus 1pt minus 1pt} 
  {\thmname{#1}\thmnumber{ #2} \textnormal{(\thmnote{#3})}} 

\newtheorem{theorem}{Theorem}
\newtheorem{corollary}{Corollary}
\renewcommand{\qedsymbol}{$\blacksquare$}
\xpatchcmd{\proof}{\hskip\labelsep}{\hskip5\labelsep}{}{}  
\makeatletter
\xpatchcmd{\proof}{\@addpunct{.}}{\@addpunct{:}}{}{}
\makeatother

\newcommand*{\KeepStyleUnderBrace}[1]{%
  \mathop{%
    \mathchoice
    {\underbrace{\displaystyle#1}}%
    {\underbrace{\textstyle#1}}%
    {\underbrace{\scriptstyle#1}}%
    {\underbrace{\scriptscriptstyle#1}}%
  }\limits
}

\newcommand*{\KeepStyleOverBrace}[1]{%
  \mathop{%
    \mathchoice
    {\overbrace{\displaystyle#1}}%
    {\overbrace{\textstyle#1}}%
    {\overbrace{\scriptstyle#1}}%
    {\overbrace{\scriptscriptstyle#1}}%
  }\limits
}

\renewcommand\[{\begin{equation}}
\renewcommand\]{\end{equation}} 
\newcounter{MYtempeqncnt}
\usepackage{listings}
\usepackage{fancyvrb}
\usepackage{framed}

\usepackage{courier}
\usepackage[usenames,dvipsnames,table]{xcolor}

\definecolor{dkgreen}{rgb}{0,0.3,0}
\definecolor{gray}{rgb}{0.5,0.5,0.5}

\makeatletter
\patchcmd{\@maketitle}
  {\addvspace{0.5\baselineskip}\egroup}
  {\addvspace{-1\baselineskip}\egroup}
  {}
  {}

\makeatother

\begin{document}

\title{Rician $K$-Factor-Based Analysis of XLOS Service Probability in 5G Outdoor Ultra-Dense Networks\footnote{This is a companion technical report of \cite{Letter}.}}

\vskip 0.4cm
\author{Hatim Chergui,~\IEEEmembership{Member,~IEEE}, Mustapha Benjillali,~\IEEEmembership{Senior Member,~IEEE},\\ and~Mohamed-Slim Alouini,~\IEEEmembership{Fellow,~IEEE}
\IEEEcompsocitemizethanks{\IEEEcompsocthanksitem H. Chergui and M. Benjillali are with the Communication Systems Department, INPT, Rabat, Morocco. [e-mail: chergui@ieee.org, benjillali@ieee.org].
\IEEEcompsocthanksitem M.-S. Alouini is with King Abdullah University of Science and Technology (KAUST), Computer, Electrical and Mathematical Science and Engineering Division (CEMSE), Thuwal 23955- 6900, Saudi Arabia. [e-mail: slim.alouini@kaust.edu.sa].}%
}

\maketitle
\thispagestyle{empty}
\begin{abstract}
In this report, we introduce the concept of Rician $K$-factor-based radio resource and mobility management for fifth generation (5G) ultra-dense networks (UDN), where the information on the gradual visibility between the new radio node B (gNB) and the user equipment (UE)---dubbed X-line-of-sight (XLOS)---would be required. We therefore start by presenting the XLOS service probability as a new performance indicator; taking into account both the UE serving and neighbor cells. By relying on a lognormal $K$-factor model, a closed-form expression of the XLOS service probability in a 5G outdoor UDN is derived in terms of the multivariate Fox H-function; wherefore we develop a GPU-enabled MATLAB routine and automate the definition of the underlying Mellin-Barnes contour via linear optimization. Residue theory is then applied to infer the relevant asymptotic behavior and show its practical implications. Finally, numerical results are provided for various network configurations, and underpinned by extensive Monte-Carlo simulations.
\end{abstract}


\newpage
\section{Introduction}
\IEEEPARstart{T}{he} emergence of massive-MIMO and millimeter-wave (mmWave) as key enablers for 5G ultra-dense networks~\cite{MIMO_mmWave} will certainly prompt the reshaping of radio resource and mobility management algorithms, wherefore a new set of measured quantities might be required as inputs.
In this context, the Rician $K$-factor can serve as an accurate channel metric to measure the gradual visibility condition of a radio link, termed  \textit{X-line-of-sight} (XLOS) here, and encompassing LOS, obstructed-LOS (OLOS) and non-LOS (NLOS) as discrete regimes. In localization services for instance, while the availability of a LOS path is quintessential for the classical triangulation-based schemes such as time-of-arrival (TOA) and direction-of-arrival (DOA), the massive-MIMO-based space-time processing approaches can deliver very concise localization thanks to the high angular resolution of the large scale antennas, and may therefore operate in the worst OLOS/NLOS conditions, yet at the expense of a higher complexity \cite{Localization}. To optimize the computational cost, an operator may adopt a hybrid network configuration where, according to a fine-tuned target $K$-factor threshold, the 5G gNB can switch between the simpler conventional methods and the massive-MIMO ones. On the other hand, in future UDNs with co-located sub-6GHz/mmWave deployment, operators might configure only an anchor internet of things (IoT) carrier on the sub-6GHz---given the scarcity of spectrum---and offload a great part of the traffic to mmWave stand-alone IoT carriers. For energy savings considerations, the 5G network might prevent IoT devices from performing the periodic inter-frequency measurements and reporting by triggering blind inter-frequency re-selections \cite{NB-IoT} to a mmWave IoT carrier once the sub-6GHz $K$-factor exceeds a preset threshold\footnote{This threshold can be determined using, e.g., machine learning approaches such as \cite{Learning}, where mmWave measurements are inferred from sub-6GHz ones.}. Zooming out from the applications, a mathematical characterization of XLOS is yet to be established.

In this letter, we propose the XLOS service probability as a performance indicator, and start by introducing a broader definition of the concept thereof; accommodating the monitoring of both the UE serving and neighbor cells. By invoking a distance-based lognormal $K$-factor model for outdoor fixed and/or sporadically moving UEs \cite{Greenstein}\footnote{The extension of the model to the vehicular case is left for future works.}, we then derive---in terms of the multivariate Fox H-function \cite[A.1]{Mathai2}---a closed-form expression for the XLOS service probability in a 5G OFDMA-based multi-tier heterogeneous network (HetNet), where a majority of cellular IoT devices and outdoor customer premise equipments (CPEs) are either fixed or changing their locations in an intermittent manner (e.g., position sensors with event-driven reporting) \cite{NB-IoT}, and where the $K$-factor variations stem also from the movement of the scattering objects (e.g., vehicles, containers, windblown leaves). Finally, the asymptotic behavior highlighting the effect of different network and channel parameters is studied using the residue theory.
\vspace{-2mm}
\section{System Model}
Consider an outdoor 2 GHz orthogonal frequency division multiple access (OFDMA)-based 5G \cite{NR_OFDM} $N$-tiers UDN, where each cell class $n$ $(n = 1,\ldots,N)$ is modeled as a homogeneous Poisson point process (PPP) $\Phi_{n}$, and distinguished by its deployment density $\lambda_{n}$, maximum transmit power per resource element (RE) $P_{n}$, antennas height $h_{n}$ and beamwidth $\theta_{n}$. The corresponding channel is presenting a large scale fading, with constant path-loss exponent $\nu$ and lognormal shadowing $\mathcal{X}_{n}$ of mean $\mu_{n}$ and standard deviation $\sigma_{n}$. Assuming that UE locations follow an independent PPP $\Phi_{u}$ of density $\lambda_{u}$, the downlink analysis is performed at a typical UE located at the origin \cite{Stochastic_Geometry_book}. To model a massive IoT and CPE device ecosystem, we further suppose that UE locations are initially fixed, but may change every now and then.
\vspace{-3mm}
\subsection{Cell Monitoring Criteria}
As we are dealing with an outdoor context, we suppose that all tier's cells are open access (including femtocells). We also adopt a reference signal receive power (RSRP)-based cell selection, wherein each UE periodically monitors the collection of the $M$ strongest cells, dubbed here \textit{monitoring set}, and ends up connecting to the best server. Since UE measurements rely on the long-term frequency-domain post-equalization receive power, small-scale fading variations do not impact cell selection/reselection and are not, therefore, reflected in the actual RSRP that reads
\begin{equation}
P_{x_{n}}=P_{n}\mathcal{X}_{n} {\|x_{n}\|}^{-\nu},\label{eq:Rx_power}
\end{equation}
where ${\|x_{n}\|}^{-\nu}$ stands for the standard path-loss between a typical UE and an $n^{\text{th}}$-tier BS located at ${x_{n}} \in \Phi_n$.
\vspace{-1.5mm}
\vspace{-3mm}
\subsection{$K$-Factor Model}
The $K$-factor--like all large scale parameters (LSPs)--follows a lognormal distribution with mean and variance depending on the frequency band, environment and  transmission/reception schemes (cf. \cite{TR38901}, \cite{MeanVar} and references therein). Without loss of generality, let us adopt the findings of \cite{Greenstein} for instance, where we assume that the narrowband $K$-factor periodically measured by a UE at independent positions can be empirically modeled for the $n\textsuperscript{th}$-tier as,
\begin{equation}
K_{x_{n}} = K_{n}\gamma_{n}{\|x_{n}\|}^{-\alpha},
\end{equation}
where $\alpha>0$, $K_{n}$ is the $K$-factor intercept defined as\footnote{According to \cite{Greenstein}, this model involves also a seasonal factor $F_s$ that reflects the vegetation. For the sake of simplicity and without loss of generality, we consider the Summer's dense vegetation case $F_s=1$.}
\begin{equation}
K_{n}=\left(h_n/h_0\right)^{\kappa_1} \left(\theta_n/\theta_0\right)^{\kappa_2} K_0,
\end{equation}
with $\kappa_1>0$, $\kappa_2<0$, $K_0>0$, and $\gamma_{n}$ is an independent lognormal variable, whose decibel value is zero mean with a standard deviation $\sigma_K$. Accurate values of these model parameters can be obtained through a calibration process according to the target environment. New Jersey's measurement campaign in \cite{Greenstein}, for instance, yields $h_0=3\, \mathrm{m}$, $\theta_0=17\degree$, $\alpha=0.5$, $\kappa_1=0.46$, $\kappa_2=-0.62$, $K_0 = 10$, and $\sigma_{K}=8 \, \mathrm{dB}$.
\subsection{Equivalent Formulation}
Since manipulating distances in PPPs is easier, let us transform the RSRP process (\ref{eq:Rx_power}) into a simple unit-power PPP $\widetilde{\Phi}_n$, where the strongest power would correspond to the nearest neighbor cell to the typical UE. By invoking the random displacement theorem \cite[1.3.9]{Stochastic_Geometry_book}, \cite[Corollary 3]{Shotgun} shows that the two-dimensional (2D) process (\ref{eq:Rx_power}) is equivalent to another 2D process
$P_{y_{n}}={\|y_{n}\|}^{-\nu}$, such that $y_{n}\in \widetilde{\Phi}_n$ with density $\widetilde{\lambda}_n=\lambda_n \Omega_n$, where $\Omega_n=P_n^{2/\nu}\mathbf{E}\left[\mathcal{X}_{n}^{2/\nu}\right]$ and the finite lognormal fractional moment $\mathbf{E}\left[\mathcal{X}_{n}^{2/\nu}\right] = \exp\left[\frac{\mathrm{ln}10}{5}\frac{\mu_n}{\nu}+\frac{1}{2}\left(\frac{\mathrm{ln}10}{5}\frac{\sigma_n}{\nu}\right)^2\right]$. By means of the mapping theorem \cite[1.3.11]{Stochastic_Geometry_book}, the $K$-factor can also be re-expressed as
\begin{equation}
K_{y_{n}} = K_{n}\gamma_{n}\Omega_{n}^{-\alpha/2} {\|y_{n}\|}^{-\alpha}, \, y_{n}\in \widetilde{\Phi}_n.
\end{equation}
\section{XLOS Service Probability}
XLOS service probability in the vicinity of a UE, $P_{\mathrm{XLOS}}$, is defined as the probability that at least one cell in the monitoring set presents a $K$-factor higher than a threshold, say $K_{\mathrm{th}}$, that can be fine-tuned depending on the target service, i.e.,
\begin{equation}
P_{\mathrm{XLOS}}\left(K_{\mathrm{th}}\right)\triangleq \mathbf{Pr}\left[\bigcup_{m=1}^{M} K_{y_{n_m}}>K_{\mathrm{th}}, \mathbf{n} \in \mathcal{M}\right],\label{eq:PLOS}
\end{equation}
where $\mathbf{n}=\left(n_1,\ldots,n_M\right)$ and $\mathcal{M}=\{1,\ldots,N \}^M $. In the sequel, we derive a closed-form expression for the XLOS service probability and study its asymptotic behavior.
\vspace{-1.5mm}
\subsection{Closed-Form Analysis}
Using the total probability theorem as well as the independence between $\gamma_{n_m},\, m=1,\ldots,M$, the definition (\ref{eq:PLOS}) can be rewritten as \vspace{-2mm}
\begin{equation}
\begin{split}
P_{\mathrm{XLOS}}\left(K_{\mathrm{th}}\right) &= 1-\mathbf{Pr}\left[\bigcap_{m=1}^{M} K_{y_{n_m}}\leq K_{\mathrm{th}}, \mathbf{n} \in \mathcal{M}\right]\\
&=1-\sum_{\mathbf{n}\in\mathcal{M}}\mathbf{Pr}\left[y_{n_m}\in \widetilde{\Phi}_{n_m},m=1,\ldots,M\right]\\
&\hspace{16mm}\times\bigintsss_{0}^{z_{n_2}}\hspace{-6mm}\ldots\hspace{-1mm}\bigintsss_{0}^{z_{n_M}}\hspace{-3mm}\bigintsss_{0}^{+\infty}\hspace{-2mm}\prod_{m=1}^{M}\hspace{-0.5mm}\mathrm{CDF}_{\gamma_{n_m}}\hspace{-1mm}\left(\hspace{-0.5mm}\frac{K_{\mathrm{th}}\Omega_{n_m}^{\alpha/2}z_{n_m}}{K_{n_m}}\Bigg\vert z_{n_m}\hspace{-1mm} \right) f\left(z_{n_1},\ldots,z_{n_M}\right)\mathrm{d}z_{n_1}\ldots\mathrm{d}z_{n_M},
\label{eq:PLOS2}
\end{split}
\end{equation}
where $z_{n_m}={\|y_{n_m}\|}^{\alpha}$ and $f\left(\cdot\right)$ is the joint probability density function (PDF) whose variables verify $0\leq z_{n_1}\leq z_{n_2}\leq\ldots\leq z_{n_M}$.
Moreover, the independence between the homogeneous PPPs $\widetilde{\Phi}_{n_m}$ as well as the superposition theorem \cite[1.3.3]{Stochastic_Geometry_book} imply that the sampling probability $\mathbf{Pr}\left[y_{n_m}\in \widetilde{\Phi}_{n_m},m=1,\ldots,M\right]=\prod_{m=1}^{M}\rho_{n_m}$, where $\rho_{n_m}=\widetilde{\lambda}_{n_m}/\lambda_T$ and $\lambda_{T}=\sum_{n=1}^{N}\widetilde{\lambda}_{n}$.
To further develop (\ref{eq:PLOS2}), let us introduce the following new theorem.

\begin{theorem}[Unified Expression for the Product of Lognormal CDFs\footnote{This theorem can be viewed as a generalization of the well-established Gauss-Hermite representations of the lognormal PDF and CDF (see e.g., \cite{Unified}).}]\label{Thm1}
Consider $M$ independent lognormal random variables $\gamma_m \, \left(m=1,\ldots,M\right)$, with mean $\mu_m(\mathrm{dB})$ and standard deviation $\sigma_m(\mathrm{dB})$. A unified expression for the product of their individual CDFs---that is equal to their joint CDF $\mathrm{CDF}_{\gamma_1,\ldots,\gamma_M}\left(\gamma_{\mathrm{th},1},\ldots,\gamma_{\mathrm{th},M}\right)$---is given by
\begin{equation}
\prod_{m=1}^{M}\hspace{-1mm}\mathrm{CDF}_{\gamma_{m}}\left(\gamma_{\mathrm{th},m}\right)=\frac{1}{\pi^{M/2}}\sum_{l=1}^{L} w_l \prod_{m=1}^{M} \mathrm{H}_{1,1}^{0,1}\left[\frac{\gamma_{\mathrm{th},m}}{\omega_{l,m}}\begin{array}{|c} (1,1)\\ (0,1) \end{array}\hspace{-1mm}\right],\label{Thm1_CDF}
\end{equation}
where $\omega_{l,m}=10^{(\sqrt{2}\sigma_m u_{l,m}+\mu_m)/ 10}$ for $l \in \{1,\ldots,L \}$, $w_l$ and $\left(u_{l,1},\ldots,u_{l,M}\right)$ are respectively the weight and the $M$ abscissas of the $L^{\mathit{th}}$-order $M$-dimensional Gaussian weight Stroud monomial cubature \cite{Stroud,Cools}, with $\sum_{l=1}^{L}w_l=\pi^{M/2}$.
\end{theorem}
\begin{proof}
cf. Appendix A.
\end{proof}

On the other hand, an explicit expression of the joint PDF $f(\cdot)$ can be obtained via the following corollary.
\begin{corollary}[of Theorem {\cite[Appendix]{Distance}}]\label{Cor1}In a multi-tier random network modeled in terms of $N$ independent PPPs $\widetilde{\Phi}_n\left(n=1,\ldots,N\right)$ with densities $\widetilde{\lambda}_n$, let $z_{m}=r_{m}^{\alpha}\left(m=1,\ldots,M\right)$, such that $r_{m}$ is the distance of the $m^{th}$ neighbor with respect to a certain origin. The joint PDF of $z_{1},\ldots,z_{M}$ unconditionally to $\{\widetilde{\Phi}_n\}$ reads
\vspace{-1.5mm}
\begin{equation}
f\left(z_{1},\ldots,z_{M}\right)=\left(\frac{2\pi\lambda_{T}}{\alpha}\right)^{\hspace{-1mm}M}\hspace{-1mm}e^{-\pi\lambda_{T}z_{M}^{2/\alpha}}\prod_{m=1}^{M}z_{m}^{2/\alpha-1},\label{eq:fz}
\vspace{-3mm}
\end{equation}
where $\lambda_{T}=\sum_{n=1}^{N}\widetilde{\lambda}_{n}$.
\end{corollary}
\begin{proof}
cf. Appendix B.
\end{proof}

By making use of the aforementioned sampling probability as well as Theorem \ref{Thm1} and Corollary \ref{Cor1}, the XLOS service probability (\ref{eq:PLOS2}) can be rewritten after some algebraic manipulations as,
\begin{equation}
\begin{split}
P_{\mathrm{XLOS}}\left(K_{\mathrm{th}}\right) &= 1-\left(\frac{2\sqrt{\pi}}{\alpha}\right)^{\hspace{-1mm}M}\hspace{-2mm}\sum_{\mathbf{n}\in\mathcal{M}}\prod_{m=1}^{M}\widetilde{\lambda}_{n_m}\sum_{l=1}^{L} w_l \times I_1, \label{eq:PLOS3}
\end{split}
\end{equation}
where the multidimensional integral $I_1$ is expressed as 
\begin{equation}
\begin{split}
I_1 &=\hspace{-1mm}\bigintsss_{0}^{z_{n_2}}\hspace{-6mm}\ldots\hspace{-1mm}\bigintsss_{0}^{z_{n_M}}\hspace{-3mm}\bigintsss_{0}^{+\infty}\hspace{-2mm}\prod_{m=1}^{M} z_{n_m}^{2/\alpha-1} \mathrm{H}_{1,1}^{0,1}\left[\frac{K_{\mathrm{th}}\Omega_{n_m}^{\alpha/2}z_{n_m}}{\omega_{l,m} K_{n_m}}\begin{array}{|c} (1,1)\\ (0,1) \end{array}\hspace{-1mm}\right]\times e^{-\pi\lambda_{T}z_{n_M}^{2/\alpha}}\,\mathrm{d}z_{n_1}\ldots\mathrm{d}z_{n_M}.\label{eq:I1}
\end{split}
\vspace{-1.5mm}
\end{equation}
\begin{figure*}[!t] 
\normalsize 
\setcounter{MYtempeqncnt}{\value{equation}} 
\setcounter{equation}{14} 
\vskip -0.4cm

\begin{equation}
\footnotesize
I_1 = \frac{\alpha}{2}\left(\pi\lambda_T\right)^{\hspace{-0.5mm}-M}\mathrm{H}_{M,M-1\colon\KeepStyleUnderBrace{ 1,1\colon\ldots\colon 1,1}_{M\mathrm{-times}}}^{\hspace{1.6mm}0,M\hspace{3.4mm}\colon\hspace{0.5mm} 0,1\colon\ldots\colon 0,1}\left[\begin{matrix}\frac{K_{\mathrm{th}}\Lambda_{n_1}^{\alpha/2}}{\omega_{l,1} K_{n_1}}\\ \vdots \\ \frac{K_{\mathrm{th}}\Lambda_{n_M}^{\alpha/2}}{\omega_{l,M} K_{n_M}}\end{matrix}\hspace{-1mm}\,\, \middle\vert \,\begin{matrix}\left(1-\frac{2i}{\alpha};\mathds{1}_{1 \leq i},\ldots,\mathds{1}_{M \leq i}\right)_{1\leq i \leq M-1},\left(1-M;\KeepStyleOverBrace{\frac{\alpha}{2},\ldots,\frac{\alpha}{2}}^{M-\mathrm{times}}\right)\\ \left(-\frac{2i}{\alpha};\mathds{1}_{1 \leq i},\ldots,\mathds{1}_{M \leq i}\right)_{1\leq i \leq M-1} \end{matrix}\hspace{-1mm}\, \middle\vert \,\begin{matrix}\left(1,1\right)\\ \left(0,1\right)\end{matrix}\, \middle\vert \,\begin{matrix}\ldots\end{matrix}\, \middle\vert \,\begin{matrix}\left(1,1\right)\\ \left(0,1\right)\end{matrix}\right]\label{eq:NFOxH}
\vspace{-3.25mm}
\normalsize
\end{equation} 

\setcounter{equation}{\value{MYtempeqncnt}} 
\hrulefill 
\vspace{-5mm}
\end{figure*}
To derive a closed-form solution for (\ref{eq:I1}), let us recall the representation of the involved Fox H-functions in terms of Mellin-Barnes integrals \cite[Eq. (1.1.1)]{Mathai2}, i.e,
\begin{equation}
\mathrm{H}_{1,1}^{0,1}\left[z\,\begin{array}{|c} (1,1)\\ (0,1) \end{array}\hspace{-1mm}\right]=\frac{1}{2\pi j}\bigintsss_{\mathcal{C}_m}\phi\left(\zeta_m\right)z^{\zeta_m} \, \mathrm{d}\zeta_m,\label{eq:Mellin}
\end{equation}
where $\phi\left(\zeta_m\right)=\Gamma\left(\zeta_m\right)/\Gamma\left(1+\zeta_m\right)$, and contours $\mathcal{C}_m \,\left(m=1,\ldots,M\right)$ are defined such that $\Re\left(\zeta_m\right)>0$; the highest pole on the left.
Combining (\ref{eq:Mellin}) with (\ref{eq:I1}) and interchanging the order of the real and contour integrals\footnote{Which is permissible given the absolute convergence of the
involved integrals.} yields
\vspace{-2mm}
\begin{equation}
\begin{split}
I_1 &=\hspace{-1mm}\left(\frac{1}{2\pi j}\right)^{\hspace{-1mm}M}\hspace{-1mm}\bigintsss_{\mathcal{C}_1}\hspace{-2.5mm}\ldots\hspace{-1.5mm}\bigintsss_{\mathcal{C}_M}\Psi\left(\zeta_1,\ldots,\zeta_M\right)\times \prod_{m=1}^{M} \phi\left(\zeta_m\right)\left(\frac{K_{\mathrm{th}}\Omega_{n_m}^{\alpha/2}}{\omega_{l,m} K_{n_m}}\right)^{\hspace{-1mm}\zeta_m} \, \mathrm{d}\zeta_1\ldots\mathrm{d}\zeta_M,
\end{split}\label{eq:I1_Mellin}
\end{equation}
with the multivariate term $\Psi$ given by
\begin{equation}
\begin{split}
\Psi\left(\zeta_1,\ldots,\zeta_M\right)&=\bigintsss_{0}^{z_{n_2}}\hspace{-6mm}\ldots\hspace{-1mm}\bigintsss_{0}^{z_{n_M}}\hspace{-3mm}\bigintsss_{0}^{+\infty}\hspace{-2mm}e^{-\pi\lambda_{T}z_{n_M}^{2/\alpha}}\times\prod_{m=1}^{M} z_{n_m}^{2/\alpha+\zeta_{m}-1}\,\mathrm{d}z_{n_1}\ldots\mathrm{d}z_{n_M}.\label{eq:Psi}
\end{split}
\end{equation}
Given that $\alpha \in \mathbb{R}^{+*}$ and $\Re\left(\zeta_m\right)>0$, and using the identity $1/a=\Gamma\left(a\right)/\Gamma\left(1+a\right)$, the iterated integrals with respect to $z_{n_1},\ldots,z_{n_{M-1}}$ in (\ref{eq:Psi}) can be successively resolved by induction. The resulting integral relating to $z_{n_M}$ is then obtained using \cite[Eq. (3.478.1)]{Table_of_Integrals}, which leads to
\hspace{-1.5mm}
\begin{equation}
\begin{split}
\Psi\left(\zeta_1,\ldots,\zeta_M\right)&=\frac{\alpha}{2}\left(\pi\lambda_T\right)^{-\left(M+\frac{\alpha}{2}\sum_{m=1}^{M}\zeta_m\right)}\times \Gamma\left(\hspace{-0.5mm} M+\frac{\alpha}{2}\sum_{m=1}^{M}\zeta_m\hspace{-1mm}\right)\hspace{-1mm}\prod_{i=1}^{M-1}\hspace{-1mm}\frac{\Gamma\left(\frac{2i}{\alpha}+\sum_{m=1}^{M}\mathds{1}_{m \leq i}\zeta_m\right)}{\Gamma\left(1+\frac{2i}{\alpha}+\sum_{m=1}^{M}\mathds{1}_{m \leq i}\zeta_m\right)}\label{eq:Psi2}.
\end{split}
\hspace{-1.5mm}
\end{equation}
By plugging (\ref{eq:Psi2}) into (\ref{eq:I1_Mellin}), we recognize that integral $I_1$ can be re-expressed in terms of the multivariate Fox H-function \cite[A.1]{Mathai2} as given by~\eqref{eq:NFOxH} on top of this page, where parameter $\Lambda_{n_m}\triangleq\Omega_{n_m}/\pi\lambda_T$ is encompassing network density, power and shadowing effects. Finally, a closed-form expression for $P_{\mathrm{XLOS}}$ is deduced by substituting (\ref{eq:NFOxH}) in (\ref{eq:PLOS3}).
\setcounter{equation}{15}
\vspace{-3.5mm}
\subsection{Asymptotic Behavior}
As depicted in Table I, the two asymptotic regimes of the ratio $K_{\mathrm{th}}\Lambda_{n_m}^{\alpha/2}/\omega_{l,m} K_{n_m}$ reflect many practical scenarios, wherefore it is interesting to establish the corresponding XLOS service probability expressions; denoted $\overline{P}_{\mathrm{XLOS}}$ in the sequel. 
Let $\mathcal{H}$ stand for the multivariate Fox H-function in (\ref{eq:NFOxH}) where
\begin{equation}
\mathcal{H}=\left(\frac{1}{2\pi j}\right)^{\hspace{-1mm}M}\hspace{-1mm}\bigintsss_{\mathcal{C}_1}\hspace{-2.5mm}\ldots\hspace{-1.5mm}\bigintsss_{\mathcal{C}_M}F\left(\zeta_1,\ldots,\zeta_M\right)\, \mathrm{d}\zeta_1\ldots\mathrm{d}\zeta_M.\label{eq:H_Mellin}
\end{equation}
In view of the series representations of the monovariate Fox H-function \cite[Theorem 1.2]{Kilbas} (while noticing the inverted definition of the H-function therein), an asymptotic expression of (\ref{eq:NFOxH}) is obtained as follows.

\noindent
\textbf{Low ratio regime}: Since the integrand $F$ has no poles on the right of the $M$ individual contours in (\ref{eq:H_Mellin}), \cite[Eq. (1.2.23)]{Kilbas} implies that $\mathcal{H}\simeq 0$, and thereby $\overline{P}_{\mathrm{XLOS}}=1$.

\noindent
\textbf{High ratio regime}: By applying \cite[Eq. (1.2.22)]{Kilbas} to the $M$ individual contour integrals, an approximation of $\mathcal{H}$ is given in terms of the residues of $F$ as
\begin{equation}
\begin{split}
\mathcal{H}&\simeq \mathbf{Res}\left[F,\left(0,\ldots,0\right)\right]+\mathbf{Res}\left[F,\left(-\frac{2}{\alpha},0,\ldots,0\right)\right]\\[-2mm]
&\simeq\lim_{\zeta_M\rightarrow0}\ldots\lim_{\zeta_1\rightarrow0} \prod_{m=1}^{M}\zeta_m F\left(\zeta_1,\ldots,\zeta_M\right)+\lim_{\zeta_M\rightarrow0}\ldots\lim_{\zeta_2\rightarrow0}\lim_{\zeta_1\rightarrow -\frac{2}{\alpha}} \left(\zeta_1+\frac{2}{\alpha}\right)\hspace{-1mm}\prod_{m=2}^{M}\zeta_m F\left(\zeta_1,\ldots,\zeta_M\right),
\end{split}
\vspace{-1.5mm}
\end{equation}
which evaluates to
\vspace{-2mm}
\begin{equation}
\mathcal{H}\simeq \left(\frac{\alpha}{2}\right)^{M-1}\left[1-\left(\frac{K_{\mathrm{th}}\Lambda_{n_1}^{\alpha/2}}{\omega_{l,1} K_{n_1}}\right)^{\hspace{-1mm}2/\alpha}\right].\label{eq:H_asym}
\vspace{-1.5mm}
\end{equation}
Finally, combining (\ref{eq:PLOS3}), (\ref{eq:NFOxH}) and (\ref{eq:H_asym}), as well as recalling that $\sum_{l=1}^{L}w_l=\pi^{M/2}$, we obtain after some algebraic manipulations
\vspace{-2mm}
\begin{equation} 
\overline{P}_{\mathrm{XLOS}}=\frac{1}{\pi^{M/2-1}}\sum_{\mathbf{n}\in \mathcal{M}}\frac{\lambda_T}{\Omega_{n_1}}\prod_{m=1}^{M}\rho_{n_m}\sum_{l=1}^{L}w_l \left(\frac{\omega_{l,1} K_{n_1}}{K_{\mathrm{th}}}\right)^{\hspace{-1mm}2/\alpha}.
\label{eq:PLOS_asym}
\end{equation} 

\begin{table}[!htb]
\vspace{-4mm}
\label{Table1}
\centering	
\newcolumntype{M}[1]{>{\centering\arraybackslash}m{#1}}

\caption{XLOS service probability asymptotic expressions}\vspace{-2mm}
\begin{tabular}{|m{4cm}|M{8cm}|M{2cm}|}
\hline 
\multicolumn{1}{|>{\centering\arraybackslash}M{4cm}|}{\cellcolor{black!20} Case} & \multicolumn{1}{>{\centering\arraybackslash}M{8cm}|}{\cellcolor{black!20} Practical Scenarios} & \multicolumn{1}{>{\centering\arraybackslash}M{2cm}|}{\cellcolor{black!20}$\overline{P}_{\mathrm{XLOS}}$}\\
\hline 
\hline
$\frac{K_{\mathrm{th}}\Lambda_{n_m}^{\alpha/2}}{\omega_{l,m} K_{n_m}}\rightarrow 0$ & 
\begin{tabular}{@{\textbullet~}p{7cm}@{}}
Poor LOS	 (i.e., low $K_{\mathrm{th}}$),\\
Fair LOS quality in a UDN with high power and narrow-beam antennas (i.e., high $\lambda_{T}$ and $K_{n_m}$).
\end{tabular} & $1$\\
\hline 
$\frac{K_{\mathrm{th}}\Lambda_{n_m}^{\alpha/2}}{\omega_{l,m} K_{n_m}}\rightarrow +\infty$ &
\begin{tabular}{@{\textbullet~}p{7cm}@{}}
High quality LOS (i.e., high $K_{\mathrm{th}}$),\\
Fair LOS quality in a low density HetNet with low power and large beamwidth antennas (i.e., low $\lambda_{T}$ and $K_{n_m}$).
\end{tabular} & Equation (\ref{eq:PLOS_asym})\\
\hline
\end{tabular}
\end{table}
\section{Numerical Results and Mathematical Software}
To validate our theoretical findings, we conduct Monte-Carlo simulations for three practical scenarios as depicted in Table~II, and we adopt New Jersey's calibration presented in II-B with $\sigma_K=3\, \mathrm{dB}$. The analytical expressions are evaluated via a degree-$11$ Stroud cubature for which $L=(4M^5-20M^4+140M^3-130M^2+96M+15)/15$. To that end, we make use of Stenger's tabulations \cite{Stenger} to update the Matlab code in \cite{StroudMatlab}. Moreover, using the quasi Monte-Carlo framework, we introduce in Appendix C an efficient GPU-oriented MATLAB routine to calculate the multivariate Fox H-function. By translating the Mellin-Barnes contour constraints into a linear optimization problem, we come up with a code automating the contour definition in Appendix D. A test example is finally provided in Appendix E.
Note that we have already introduced a C/MEX version of the multivariate Fox H-function in our package \cite{mfoxh}. An excerpt of the source and test examples are presented in Appendices F and G, respectively. Also, a Python implementation for the same generalized function can be found in \cite{Python}.

\begin{figure}[t!]
\hspace{5cm}
\includegraphics[scale=0.6]{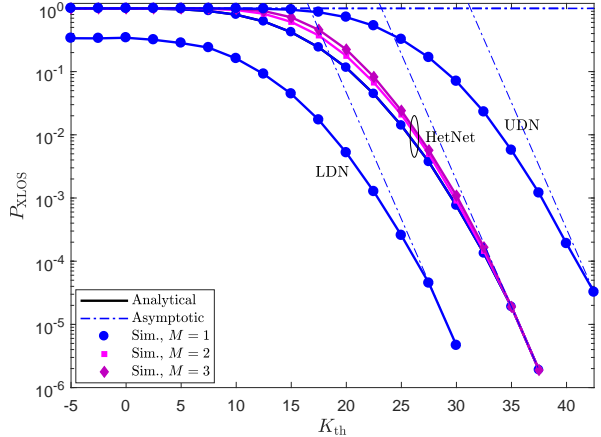}
\caption{XLOS probability versus $K_{\mathrm{th}}$ for UDN, HetNet (Macro/Femto) and LDN. Path-loss exponent $\nu=3$ and shadowing mean $\mu_n=0$ for all tiers $n=1,\ldots,N$.}
\end{figure}

Fig.~1 shows that, in the UDN case, the low ratio asymptotic regime is easily established, and good LOS conditions (e.g., $K_{\mathrm{th}}=12\,\mathrm{dB}$) are obtained with probability 1. Conversely, the low density network (LDN) scenario unfolds in NLOS situations with non-negligible probability (e.g., $K<-5\,\mathrm{dB}$ with probability $0.5$). By considering the neighboring cells ($M=2,3$) in the HetNet (Macro/Femto) case for instance, we remark that a substantial increase of the XLOS probability is achieved only in the non-asymptotic regime. Indeed, a high $K_{\mathrm{th}}$ requirement can be fulfilled merely by the serving cell, since the $K$-factors of neighbor cells become limited by the corresponding path-losses. 
\begin{table}[!htb]
\vspace{-4mm}
\label{Table1}
\centering
\newcolumntype{M}[1]{>{\centering\arraybackslash}m{#1}}

\caption{Network and transmission settings}\vspace{-2mm}
\begin{tabular}{|m{1cm}|m{0.25cm}|M{3cm}|M{1.5cm}|M{1cm}|M{1cm}|M{1cm}|}
\hline 
\multicolumn{1}{|>{\centering\arraybackslash}M{1cm}|}{\cellcolor{black!20} Case} & \multicolumn{1}{|>{\centering\arraybackslash}M{0.25cm}|}{\cellcolor{black!20} $N$} & \multicolumn{1}{>{\centering\arraybackslash}M{3cm}|}{\cellcolor{black!20} $\lambda_n$}& \multicolumn{1}{>{\centering\arraybackslash}M{1.5cm}|}{\cellcolor{black!20} $P_n(\mathrm{dBm})$}& \multicolumn{1}{>{\centering\arraybackslash}M{1cm}|}{\cellcolor{black!20} $\theta_n(\degree)$}& \multicolumn{1}{>{\centering\arraybackslash}M{1cm}|}{\cellcolor{black!20} $h_n(\mathrm{m})$}& \multicolumn{1}{>{\centering\arraybackslash}M{1cm}|}{\cellcolor{black!20} $\sigma_n (\mathrm{dB})$}\\
\hline 
\hline
UDN &$1$&$3\times10^{-2}$&$5.2$&$45$&$15$&5\\
\hline 
HetNet&$2$&$10^{-5},5\times10^{-4}$&$15.2, -4.8$&$39,180$&$25, 10$&8, 4\\
\hline
LDN &$1$&$3\times10^{-8}$&$5.2$&$65$&$30$&5\\
\hline
\end{tabular}
\end{table}
\vspace{-0.5cm}
\section{Conclusion}
In this letter, we have introduced the XLOS service probability as a new $K$-factor-based performance indicator, and provided its analytical and asymptotic expressions that establish a link between 5G network and transmission parameters and the gradual visibility condition of radio links. As exemplified in the introduction, by tweaking a $K$-factor threshold $K_{\mathrm{th}}$, the XLOS metric can be used by network optimization engineers as a switching probability between e.g., conventional localization schemes---requiring LOS---and Massive-MIMO ones operating in OLOS conditions. As a perspective, the adopted $K$-factor model from \cite{Greenstein} can be extended to the vehicular case in future works.
\vspace{-3mm}
\appendices
\section{Proof of Theorem 1}
First, by making a simple variable change, the product of lognormal PDFs $p_{\gamma_m},\ p=\prod_{m=1}^{M}p_{\gamma_m}$, can be reformulated as
\begin{equation}
\vspace{-1mm}
p\hspace{-0.5mm}=\hspace{-0.5mm}\frac{1}{\pi^{M/2}}\hspace{-1.2mm}\bigintssss_{\mathbb{R}^{M}}\hspace{-3.2mm}e^{-\left(u_{1}^{2}+\ldots+u_{M}^{2}\right)}Q\left(u_1,\ldots,u_M\right)\mathrm{d}u_1\ldots\mathrm{d}u_M, \label{Proof_PDF}
\vspace{-1mm}
\end{equation}
where $Q\left(u_1,\ldots,u_M\right)=\prod_{m=1}^{M}\delta\left(\gamma_m-10^{(\sqrt{2}\sigma_m u_{m}+\mu_m)/ 10}\right)$. By applying the Gaussian-weight Stroud monomial cubature \cite{Stroud,Cools} to (\ref{Proof_PDF}), and recalling that $\delta\left(\gamma_m-a\right)=\mathrm{H}_{0,0}^{0,0}\left[\frac{\gamma_m}{a}\begin{array}{|c} \rule{5mm}{0.4pt}\\ \rule{5mm}{0.4pt} \end{array}\hspace{-1mm}\right]$, we get
\begin{equation}
\vspace{-1mm}
p=\frac{1}{\pi^{M/2}}\sum_{l=1}^{L}w_l\prod_{m=1}^{M}\mathrm{H}_{0,0}^{0,0}\left[\frac{\gamma_m}{\omega_{l,m}}\begin{array}{|c} \rule{5mm}{0.4pt}\\ \rule{5mm}{0.4pt} \end{array}\hspace{-1mm}\right],\label{Proof_PDF2}
\vspace{-1mm}
\end{equation}
with $w_l$ and $\left(u_{l,1},\ldots,u_{l,M}\right)$ are respectively the $l^{\mathrm{th}}$ weight and abscissas of the $M$-dimensional cubature, and $\omega_{l,m}=10^{(\sqrt{2}\sigma_m u_{l,m}+\mu_m)/ 10}$. Finally, by invoking \cite[Eq. (2.53)]{Mathai2}, the integration of (\ref{Proof_PDF2}) with respect to $\gamma_m$ from $0$ to $\gamma_{\mathrm{th,m}}\,(m=1,\ldots,M)$ leads to (\ref{Thm1_CDF}).\hfill\qedsymbol
\vspace{-3mm}
\section{Proof of Corollary 1}
It immediately follows from applying the superposition theorem \cite[1.3.3]{Stochastic_Geometry_book} to the equivalent PPP $\Phi_T=\bigcup_{n=1}^{N}\widetilde{\Phi}_n$, and performing a PDF transformation to the joint distance distribution of the first $M$ neighbors given by theorem \cite[Appendix]{Distance}.\hfill\qedsymbol
\newpage
\section{GPU-Enabled Multivariate Fox H-Function MATLAB Code}
\definecolor{listinggray}{gray}{0.9}
\definecolor{lbcolor}{rgb}{0.97,0.97,0.97}
\lstset{
tabsize=4,    
language=Matlab,    
keywords={break,case,catch,continue,else,elseif,end,for,function,global,if,otherwise,persistent,return,switch,try,while},    basicstyle=\scriptsize\ttfamily, 
keywordstyle=\color{blue},    
commentstyle=\color{dkgreen}, 
stringstyle=\color{Purple},
numbers=left,    
numberstyle=\tiny\color{gray},    
stepnumber=1,    
numbersep=10pt, 
breaklines=true,
prebreak = \raisebox{0ex}[0ex][0ex]{\ensuremath{\hookleftarrow}},
frame=single,   
backgroundcolor=\color{lbcolor},    
tabsize=4,    
showspaces=false,    
showstringspaces=false}
\begin{lstlisting}
function out = mfoxh(z, Contour, an, Alphan, ap, Alphap, bq, Betaq, varargin)
% For dim >= 4, we recommend the use of GPU-enabled HPC servers
% an = [a1,...,an] and Alphan = [alpha,1,1 ... alpha,n,1;...; alpha,1,r...alpha,n,r]                      
% varargin form (i = 1 ..r): [ci,1 ...ci,n ; gammai,1 ... gammai,n], 
%[ci,n+1 ...ci,p ; gammai,n+1 ... gammai,p], 
%[di,1 ...di,m ; deltai,1 ... deltai,m], [di,m+1 ...di,q ; deltai,m+1 ... deltai,q]
% See notation in A. Mathai, The H-function, Theory and Applications, Annex A.1
%================================================================================================
dim    = size(Contour,1);

N      = 2^8 * 3^3 * 5^(2+dim); % For a better performance, N is made of powers of prime numbers
%========================== Multivariate Quasi Monte-Carlo Integration ==========================
p      = haltonset(dim,'Skip',1e3,'Leap',1e2);
p      = scramble(p,'RR2');
in     = gpuArray(net(p,N));
G      = gpuArray(ones(1,N));
C      = gpuArray(Contour);
points = kron(G,C(:,1)) + kron(G, C(:,2)-C(:,1)).* in';
mceval = Integrand(points);
mcsum  = sum(mceval,2);
v      = prod(C(:,2) - C(:,1)); % volume
out    = v * mcsum / N; % Integral
%========================================== Integrand ===========================================
function f = Integrand(s)
j       = sqrt(-1);
r       = length(z);
Nvar    = length(varargin);
for nvar = 1 : Nvar
 if(isempty(cell2mat(varargin(nvar))))
  varargin{nvar} = zeros(2,0);
 end
end
Phi = 1;
for i = 1 : r   
 cni = gpuArray(cell2mat(varargin(4*(i-1)+1)));
 cpi = gpuArray(cell2mat(varargin(4*(i-1)+2)));
 dmi = gpuArray(cell2mat(varargin(4*(i-1)+3)));
 dqi = gpuArray(cell2mat(varargin(4*(i-1)+4)));
 Phi = Phi .* ((GammaProd(1-cni(1,:),cni(2,:), s(i,:)).* GammaProd(dmi(1,:),-dmi(2,:), s(i,:)))...
  ./(GammaProd(cpi(1,:),-cpi(2,:), s(i,:)).* GammaProd(1-dqi(1,:),dqi(2,:), s(i,:)))).* z(i).^s(i,:);
end
Psi     = GammaProd(1-an,Alphan,s)./(GammaProd(ap,-Alphap,s).* GammaProd(1-bq,Betaq,s));
f       = (1/(2*j*pi)^r) * Phi .* Psi;
end
%========================================== GammaProd ===============================================
function output = GammaProd(p,m,s) 
if (isempty(p)|| isempty(m)) 
    output = ones(size(s(1,:)));
else
L1    = size(s,1);    
comb  = 0;
for i = 1 : L1
[pp ss]  = meshgrid(p,s(i,:));
 mm      = meshgrid(m(i,:),s(i,:));
 comb    =  comb + mm .* ss;
end
    output = reshape(prod(gammas(pp + comb),2),size(s(1,:)));
end
end
end
% gammas function here is the complex gamma, available in
% www.mathworks.com/matlabcentral/fileexchange/3572-gamma
end
\end{lstlisting}

\section{Automatic Contour Generator}
\begin{lstlisting}
function c = mfoxcontour(W, dim, an, Alphan, varargin)
% an = [a1,...,an] and Alphan = [alpha,1,1 ... alpha,n,1;...; alpha,1,r...alpha,n,r]                      
% varargin form (i = 1 ..r): [ci,1 ...ci,n ; gammai,1 ... gammai,n], 
% [di,1 ...di,m ; deltai,1 ... deltai,m]
% See notation in A. Mathai, The H-function, Theory and Applications, Annex A.1
% W   : control the width of the integration interval in [-i\infty +i\infty]
% dim : stands for the dimension

Nvar    = length(varargin);
epsilon = 1/10;
f  = ones(1,dim);
Q  = -Alphan.';
b  = 1-an-epsilon;
lb = [];
ub = [];

for i = 1 : Nvar/2    
 cni = cell2mat(varargin(2*(i-1)+1)); % [c1,i ...cn,i;gamma_1,i...gamma_n,i]
 if(isempty(cni)) cni = [-1e10;1]; end
 dmi = cell2mat(varargin(2*(i-1)+2));  % [d1,i ...dm,i;delta_1,i...delta_m,i]
 if(isempty(dmi)) dmi = [1e10;1]; end
 lb  = [lb max((cni(1,:)-1)./cni(2,:))]+epsilon;
 ub  = [ub min(dmi(1,:)./dmi(2,:))]-epsilon;    
end

options = optimoptions('linprog','Algorithm','interior-point');
out = linprog(f, Q, b, [], [], lb, ub, options);
c = [out.' - 1i * W; out.' + 1i * W];    
\end{lstlisting}

\section{MATLAB-GPU Test Code}
In this example, we evaluate both a trivariate and bivariate Fox H-functions, respectively given by,
\begin{equation}
\mathrm{H}_1=\mathrm{H}_{2,1:1,0:1,0:1,0}^{0,1:0,1:0,1:0,1}\left(\hspace{-1mm}\begin{array}{c} 3, 2, 0.5\end{array}\begin{array}{|c} \left(1.5; 1, 1, 1\right), \left(2; 1, 1, 1\right)\\ \left(2; 1, 1, 1\right)\end{array}\begin{array}{|c} \_\\ \left(0,1\right) \end{array}\begin{array}{|c} \_\\ \left(3,1\right)\end{array}\begin{array}{|c} \_\\ \left(1,1\right) \end{array}\hspace{-1mm}\right),
\end{equation}
and
\begin{equation}
\mathrm{H}_2=\mathrm{H}_{2,1:1,0:1,0}^{0,1:0,1:0,1}\left(\hspace{-1mm}\begin{array}{c} 3, 2\end{array}\begin{array}{|c} \left(1.5; 1, 1\right), \left(2; 1, 1\right)\\ \left(2; 1, 1\right)\end{array}\begin{array}{|c} \_\\ \left(0,1\right) \end{array}\begin{array}{|c} \_\\ \left(3,1\right)\end{array}\hspace{-1mm}\right).
\end{equation}
The user may set the contour manually or generate it via \ttfamily{mfoxcontour} \normalfont routine provided in Appendix D.

\begin{lstlisting}
% Trivariate example
z = [3 2 0.5];
an = [1.5];
ap = [2];
Alphan = [1 ; 1 ; 1];
Alphap = [1 ; 1 ; 1];
bq = [2];
Betaq = [1 ; 1 ; 1];
Contour = mfoxcontour(10, 3, an, Alphan, [],[0;1],[],[3;1],[],[1;1]);
Contour =
  -3.1000 -10.0000i  -3.1000 +10.0000i
   2.8000 -10.0000i   2.8000 +10.0000i
   0.9000 -10.0000i   0.9000 +10.0000i
H1 = mfoxh(z, Contour, an, Alphan, ap, Alphap, bq, Betaq,[],[],[0;1],[],[],[],[3;1],[],[],[],[1;1],[])
H1 =
   0.4886 + 0.0035i
   
% Bivariate example
z = [3 2];
an = [1.5];
ap = [2];
Alphan = [1 ; 1];
Alphap = [1 ; 1];
bq = [2];
Betaq = [1 ; 1];
Contour = [-1.5-10i -1.5+10i ; 2.5-10i 2.5+10i]; % contour set manually
H2 = mfoxh(z, Contour, an, Alphan, ap, Alphap, bq, Betaq,[],[],[0;1],[],[],[],[3;1],[])
H2 =
  -0.6014 + 0.0011i
\end{lstlisting}
\section{Parallel-CPU C/MEX Code: Excerpt from our Package \cite{mfoxh}}
\definecolor{listinggray}{gray}{0.9}
\definecolor{lbcolor}{rgb}{0.97,0.97,0.97}
\lstset{
backgroundcolor=\color{lbcolor},
    tabsize=4,    
    language=[GNU]C++,
        basicstyle=\scriptsize,
        upquote=false,
        aboveskip={1.5\baselineskip},
        columns=fixed,
        showstringspaces=false,
        extendedchars=false,
        breaklines=true,
        prebreak = \raisebox{0ex}[0ex][0ex]{\ensuremath{\hookleftarrow}},
        frame=single,
        numbers=left,
        showtabs=false,
        showspaces=false,
        showstringspaces=false,
        identifierstyle=\ttfamily,
        keywordstyle=\color[rgb]{0,0,1},
        commentstyle=\color[rgb]{0.026,0.112,0.095},
        stringstyle=\color[rgb]{0.627,0.126,0.941},
        numberstyle=\color[rgb]{0.205, 0.142, 0.73},
}

\begin{lstlisting}

#include <omp.h>
#include <stdio.h>
#include <math.h>
#include <gsl/gsl_qrng.h>
#include <gsl/gsl_complex.h>
#include <gsl/gsl_complex_math.h>
#include <gsl/gsl_sf_gamma.h>
#include <gsl/gsl_blas.h>
#include <gsl/gsl_vector_complex_double.h>
#include <gsl/gsl_matrix_complex_double.h>
#include <gsl/gsl_matrix_int.h>
#include "mfox.h"
#include "mex.h"

void mexFunction(int nlhs, mxArray *plhs[],
                 int nrhs, const mxArray *prhs[]){
     
    size_t i, j, k, mx, nx, dim;
    double  *ind, *xr, *xi, *zr, *zi, *max_call, *tol;
    gsl_complex x;
    gsl_vector_complex *xl, *xu;
    gsl_matrix_int *index;
    gsl_matrix_complex *Arg[20], *Emp;

if((nrhs -7) % 2  == 0){
dim = (size_t)((nrhs -7)/2);}
else{mexErrMsgTxt("Number of inputs is incorrect\n");}

// Retrieve vector [0, n, m1, n1, ..., mM, nM]
if(mxIsComplex(prhs[0])){mexErrMsgTxt("Indices should be integers\n");}
nx = mxGetN(prhs[0]);
if(nx != 2*dim+2){mexErrMsgTxt("Missing input(s) or extra elements in the first argument (index)\n");}
index   = gsl_matrix_int_alloc(1, nx);
ind     = mxGetPr(prhs[0]);
for(j = 0; j < nx; j++){
gsl_matrix_int_set(index, 0, (const size_t)j, (int)ind[j]);
}
// Retrieve matrices Ai and Bi

for(k = 0; k < nrhs-4; k++){
if(! mxIsEmpty(prhs[k+1]))
{
  /* Get the length of each input vector. */
  mx = mxGetM(prhs[k+1]);
  nx = mxGetN(prhs[k+1]);
  /* Check input parameters size */

if(k == 0 && nx != dim){mexErrMsgTxt("Missing Fox H argument(s) z\n");}
if(k > 0 && nx ==1){
         mexPrintf("Error in parameter # %d\n",k+2); 
         mexErrMsgTxt("Input size is incorrect\n");            
         return;
         }
if(mxIsComplex(prhs[k+1])){
  /* Get pointers to real and imaginary parts of the inputs. */
  xr = mxGetPr(prhs[k+1]);
  xi = mxGetPi(prhs[k+1]);

  Arg[k] = gsl_matrix_complex_alloc(mx, nx);

for(i = 0; i < mx; i++)
  {
      for(j = 0; j < nx; j++)
      {
          GSL_SET_COMPLEX(x,xr[i + mx*j],xi[i + mx*j]);
          gsl_matrix_complex_set(Arg[k], 
                     (const size_t)i, (const size_t)j, x);
      }
  }
}//end of if

else{
/* Get pointers to real part */
  xr = mxGetPr(prhs[k+1]);
  Arg[k] = gsl_matrix_complex_alloc(mx, nx);

for(i = 0; i < mx; i++)
  {
      for(j = 0; j < nx; j++)
      {
          GSL_SET_COMPLEX(x,xr[i + mx*j],0.0);
          gsl_matrix_complex_set(Arg[k], 
                     (const size_t)i, (const size_t)j, x);
      }
  }
}//end of else
}
else{
  Arg[k] = gsl_matrix_complex_alloc(2,1);
  gsl_matrix_complex_set_all(Arg[k],GSL_COMPLEX_ONE);
}
}// end of for(k=0...
// Retrieve integration intervals
  if(mxGetN(prhs[nrhs-3]) < dim || mxGetM(prhs[nrhs-3]) < 2){
                  mexErrMsgTxt("Contour matrix size incorrect"); 
                  }
  xl = gsl_vector_complex_alloc(dim);
  xu = gsl_vector_complex_alloc(dim);
/* Get pointers to real and imaginary parts of the inputs. */
  xr = mxGetPr(prhs[nrhs-3]);
  xi = mxGetPi(prhs[nrhs-3]);
//Initialize integration domains
for (i = 0; i < dim; i++) {
  gsl_vector_complex_set (xl, i, gsl_complex_rect(xr[2*i],xi[2*i]));
  gsl_vector_complex_set (xu, i, gsl_complex_rect(xr[2*i+1],xi[2*i+1]));
}
// Retrive max_call and tolerence
if(!mxIsComplex(prhs[nrhs-2]) && !mxIsComplex(prhs[nrhs-1])){
   max_call = mxGetPr(prhs[nrhs-2]);
   tol      = mxGetPr(prhs[nrhs-1]);
  }
else{mexErrMsgTxt("MaxFunEval and AbsTol must be real");}
  
//Perform integration
  gsl_complex result = GSL_COMPLEX_ZERO, error= GSL_COMPLEX_ZERO;
  gsl_qrng* qrng = gsl_qrng_alloc(gsl_qrng_reversehalton, dim);
/*Available sequences: gsl_qrng_halton, gsl_qrng_sobol, gsl_qrng_niederreiter_2, gsl_qrng_reversehalton*/
  quasi_monte_state* s = quasi_monte_alloc(dim);
  quasi_monte_integrate(xl, xu, dim, max_call[0], 0, tol[0], qrng, s, &result, &error, index, Arg);
  quasi_monte_free(s);
  gsl_qrng_free(qrng);
/* Create a new complex array and set the output pointer to it. */
  plhs[0] = mxCreateDoubleMatrix(1, 1, mxCOMPLEX);
  zr = mxGetPr(plhs[0]);
  zi = mxGetPi(plhs[0]);
  zr[0] = GSL_REAL(result);
  zi[0] = GSL_IMAG(result);
}
\end{lstlisting}
\vspace{-5mm}
\section{MEX Test Code: Excerpt from our Package \cite{mfoxh}}
\definecolor{listinggray}{gray}{0.9}
\definecolor{lbcolor}{rgb}{0.97,0.97,0.97}
\lstset{
tabsize=4,    
language=Matlab,    
keywords={break,case,catch,continue,else,elseif,end,for,function,global,if,otherwise,persistent,return,switch,try,while},    basicstyle=\scriptsize\ttfamily, 
keywordstyle=\color{blue},    
commentstyle=\color{dkgreen}, 
stringstyle=\color{Purple},
numbers=left,    
numberstyle=\tiny\color{gray},    
stepnumber=1,    
numbersep=10pt, 
breaklines=true,
prebreak = \raisebox{0ex}[0ex][0ex]{\ensuremath{\hookleftarrow}},
frame=single,   
backgroundcolor=\color{lbcolor},    
tabsize=4,    
showspaces=false,    
showstringspaces=false}
\vspace{-6mm}
\begin{lstlisting}
%***************************************** 2-Dimensions **************************************** 
 index = [0 1 1 1 1 1]; % [0 m m1 n1 m2 n2 ...mM nM]
 z = [1 2]; % [z1...zM]
 A = [1.5 1.0 1.0; 2.0 1.0 1.0]; % [a1 alpha_1,1...alpha_1,M; ...; ap alpha_p,1...alpha_p,M]
 B = [2.0 1.0 1.0]; % [b1 beta_1,1...beta_1,M;...; bq beta_q,1 ...beta_q,M]
 A1 = [-1 1]; % [a1_1 alpha1_1;...;a1_p1 alpha1_p1]
 B1 = [0 1]; % [b1_1 beta1_1;...;b1_q1 beta1_q1]
 A2 = [-1 1];
 B2 = [3 1];
 %c  = [-0.5-10i 1.5-10i;-0.5+10i 1.5+10i];
 c = mfoxcontour(10, 2, index, A, A1, B1, A2, B2); 
 % Tolerence settings
 MaxFunEval = 2*1e5; % increase it to get more precision (especially for more than 2 dimensions)
 AbsTol     = 1e-4;
 tic;
 mfoxh(index, z, A, B, A1, B1, A2, B2, c, MaxFunEval, AbsTol)
 toc;
%***************************************** 3-Dimensions **************************************** 
 index = [0 1 1 1 1 1 1 1];
 z = [1 2 0.5];
 A = [1.5 1.0 1.0 1; 2.0 1.0 1.0 1];
 B = [2.0 1.0 1.0 1];
 A1 = [-1 1];
 B1 = [0 1];
 A2 = [-1 1];
 B2 = [3 1];
 A3 = [0 1];
 B3 = [1 1];
 %c  = [-0.5-10i 1.5-10i 0.5-10i;-0.5+10i 1.5+10i 0.5+10i];
 c = mfoxcontour(10, 3, index, A, A1, B1, A2, B2, A3, B3)
 % Tolerence settings
 MaxFunEval = 2*1e6; % increase it to get more precision (especially for more than 2 dimensions)
 AbsTol     = 1e-4;
 tic;
 mfoxh(index, z, A, B, A1, B1, A2, B2, A3, B3, c, MaxFunEval, AbsTol)
 toc;
%\end{lstlisting}

\end{document}